\documentclass[11pt,twocolumn]{IEEEtran}

\usepackage{amsmath}
\usepackage{amsfonts}
\usepackage{amssymb}
\usepackage{mathtools}
\usepackage{amsthm}
\usepackage{fontenc}
\usepackage{graphicx}
\usepackage[linesnumbered,ruled]{algorithm2e}
\usepackage{booktabs}
\usepackage{subfigure}
\usepackage{tabulary}
\usepackage{booktabs}
\usepackage[justification=centering]{caption}
\usepackage{comment}
\usepackage{cite}
\newtheorem{theorem}{\textbf{Theorem}}
\newtheorem{Prob}{\textbf{Problem}}

\newtheorem{property}{\textbf{Property}}
\SetKw{KwBy}{by}
\usepackage[colorinlistoftodos]{todonotes}
\usepackage{caption}

\begin{document}
\pagenumbering{gobble}

\title{Communications-Caching-Computing Tradeoff Analysis for Bidirectional Data Computation in Mobile Edge Networks\thanks{}}
\author{\text{Yaping Sun}, \text{Lyutianyang Zhang},\ \text{Zhiyong Chen}, \text{and Sumit Roy,}  \IEEEmembership{Fellow, IEEE}
\thanks{Yaping Sun and Zhiyong Chen are with Cooperative Medianet Innovation Center, Shanghai Jiao Tong University, Shanghai 200240, China (e-mail: \{yapingsun, zhiyong chen\}@sjtu.edu.cn). Lyutianyang Zhang and Sumit Roy are with Department of Electrical \& Computer Engineering, University of Washington, Seattle, WA, USA (e-mail:\{lyutiz,sroy\}@uw.edu).}}
\maketitle
\begin{abstract}
  With the advent of the modern mobile traffic, e.g., online gaming, augmented reality delivery and etc.,  a novel bidirectional computation task model where the input data of each task consists of two parts, one  generated at the mobile device in real-time and the other originated from the Internet proactively, is emerging as an important use case of 5G. In this paper, for ease of analytical analysis, we consider the homogeneous bidirectional computation task model in a mobile edge network which consists of one mobile edge computing (MEC) server and one mobile device, both enabled with computing and caching capabilities.
 Each task can be served via three mechanisms, i.e., local computing with local caching, local computing without local caching and computing at the MEC server. 
 To minimize the average bandwidth, we formulate the joint caching and computing optimization problem under the latency, cache size and average power constraints. We derive the closed-form expressions for the optimal policy and the minimum bandwidth. The tradeoff among communications, computing and caching is illustrated both analytically and numerically, which provides insightful guideline for the network designers. 
\end{abstract}
%\begin{IEEEkeywords}   Bandwidth Minimization, Bidirectional Computation Task, 3C Tradeoff, Mobile Edge Network.
%\end{IEEEkeywords}
%\newpage

\section{Introduction}

The advent of modern mobile traffic, e.g., online gaming, mobile virtual reality (VR)/augmented reality (AR) delivery and etc., incurs ultra-high requirements on the wireless bandwidth \cite{nature}. For example, the mobile VR delivery requires the transmission rate on the order of G bit/s \cite{highorder}. Mobile edge network (MEN) that equips the edge nodes of the mobile network, e.g., the mobile edge computing (MEC) server and the mobile devices, with caching and computing resources is deemed as one of the most promising approaches to alleviate the bandwidth burden on the mobile carriers \cite{colors}. In particular, mobile edge caching indicates proactively storing popular contents into the network edge nodes to reduce the traffic redundancy and transmission latency \cite{ali,joint}. MEC refers to computing the tasks at the network edge nodes to reduce the core network burden and the latency \cite{mec00,mec01,mec02,efficient,tra,mec_tony,liulei}.  How to efficiently utilize the caching and computing resources in  MEN triggers the research interests from both the academic and industrial areas \cite{efficient,tra,liulei,mec_tony,yang,TCOM_sun,TWC_sun}.  

The computation model in the currently existing literature on MEC can be named as \textit{one-way} computation task model. That is, the input data of each computation task is assumed to be either generated at the mobile device \cite{efficient,tra,mec_tony} or originated from the Internet \cite{yang, TCOM_sun,TWC_sun}. In particular, in \cite{efficient,tra,mec_tony}, the mobile device offloads the input data to the MEC server for computation and then downloads the output data from the MEC server. In \cite{yang,TCOM_sun,TWC_sun}, when the task is computed at the mobile device, the mobile device has to download the input data from the MEC server first if not cached locally and then computes the input data to obtain the output data. 

Novelly, in this paper, we consider a \emph{bidirectional} computation task model, where the input data consists of two parts, one of which is generated from the mobile device in real-time and the other of which is originated from the Internet proactively. One of the most directly motivating examples is online 
Role-Playing Game (RPG). Suppose one player is controlling a role and choosing which place/map to go. The location of the role combined with the map information from the MEC server could help render the picture for the player after some computations. The input data consists of these generated at the mobile device in real-time including current player equipment/weapon selection, strategy selection as well as role selection, and also those proactively generated from the Internet such as the map information. This rendering task could be done either at the mobile device or at the MEC server. If the task is computed at the MEC server, the mobile device has to first upload the player's related information to the MEC server, then the MEC server computes the task and transmits the computation result to the mobile device. If the task is computed at the mobile device, the mobile device has to first download the map information from the MEC server and then computes the rendering task. Since the required transmission load and the computation frequency when computing at the mobile device are different from those when computing at the MEC server, the corresponding consumed bandwidth  differs and thus the computing policy requires careful design. Besides, the history of all the players' actions could provide a popularity distribution of the map preferences, e.g., the maps/places the players mostly like to go to. Based on a priori knowledge of the popularity, the popular maps/places could be proactively cached at the mobile device to save the consumed bandwidth. 

%Besides, joint design of caching and computing policy is drawing increasing attention based on the fact that the input data or the output data of the computation task is cacheable \cite{yang,TCOM_sun,TWC_sun}. Still take the RPG as an example. 

Inspired by this, under the latency, cache size and average power constraints,  this paper jointly optimizes the computing and caching policy to minimize the average bandwidth for the bidirectional computation task model. Then, we derive the closed-form expressions for the optimal policy and the minimum bandwidth in the homogeneous scenario. The tradeoff among communications, computing and caching (3C) is at last illustrated both theoretically and numerically.%Resource allocation at the MEC/cloud server has been considered in \cite{caching_mec_twc,edgeCC,hao2018,chen2018edge,yu2018computation}. In particular, \cite{caching_mec_twc} designs optimal computing offloading and caching policy for minimizing the latency in a hybrid mobile cloud/edge computation system. \cite{edgeCC} studies the network-assisted video rate adaptation problem in an MEC-enabled base station with caching.  \cite{hao2018} considers the joint optimization of task caching and offloading on edge cloud to minimize the energy consumption. \cite{chen2018edge} considers the expected latency minimization for each task service under the cache size constraint in the edge cloud. \cite{yu2018computation} considers the optimal fine-grained collaborative offloading strategies with caching-enhancements at the MEC server to minimize the overall execution delay of the mobile device. Moreover, all the works in \cite{caching_mec_twc,edgeCC,hao2018,chen2018edge,yu2018computation} consider one-way computation task model in which the input data of each computation task is generated at the mobile device. On the other hand, resource allocation at the mobile devices has been considered in \cite{yang,TCOM_sun,TWC_sun}. In particular, \cite{yang,TCOM_sun} consider the joint caching and computing policy optimization to minimize the bandwidth comsumption for the mobile VR delivery in a single-user scenario. \cite{TWC_sun} considers the joint design in a multiple-user scenario by taking into account the multicast opportunity. All the works in \cite{yang,TCOM_sun,TWC_sun} consider one-way computation task model in which the input data of each computation task is originated from the Internet proactively.  

\section{System Model}

\begin{figure}[t]
  \centering
  \includegraphics[width=70mm, height=26mm]{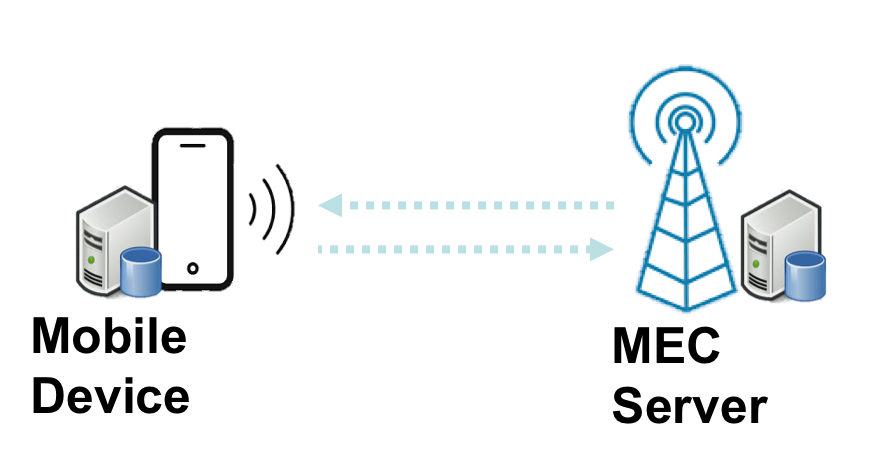}
\caption{System Model}
\label{fig:System_model}
\end{figure}

As illustrated in Fig.~\ref{fig:System_model}, we consider a mobile edge network consisting of one MEC server and one mobile device, both equipped with certain caching and computing abilities.%\footnote{\textcolor{black}{ This system model can be easily extended to a multi-user scenario.}} %In this paper, \textcolor{black}{we optimize the joint caching and computing policy at the mobile device to minimize the average bandwidth consumption under the latency, cache size, and average power constraints}. 
The mobile device is assumed to \textcolor{black}{request} one task at each time. %Key notations are summarized in Table \ref{tab:notation}.

\subsection{Task Model}
%The system involves a mobile device and a MEC server which take turns to do caching and computing cooperatively to complete a task. 
Assume that there are $F$ tasks in total to be requested by the mobile device. Denote with $\mathcal{F} \overset{\Delta}{=} \{1,2, \cdots,f,\cdots, F\}$ the task set. Each task $f\in \mathcal{F}$ is characterized by a $5$-item tuple $\Big\{I^D\ (\text{in\ bits}),\ I^S\ (\text{in\ bits}),\! O\  (\text{in\ bits}),\! w \ (\text{in cycles/bit}),\\ \tau\ (\text{in seconds})\Big\}$.\footnote{This system model can be directly extended to a multi-user heterogeneous scenario.} In particular, for each task $f\in \mathcal{F}$, $I^D$ represents the size of the local input data which is generated at the mobile device in real time. $I^S$ represents the size of the remote input data which is originated from the Internet and can be proactively stored. $O$ represents the size of the output data. $w$ and $\tau$ denote the required computation cycles per bit and the maximum tolerable service latency, respectively. Since the input remote data is generated proactively, the task popularity can be learned based on the request history information. The task request process at the mobile device is assumed to conform to the independent reference model based on the following assumptions \cite{TCOM_sun}: i) the tasks that the mobile device wants to process is fixed to the set $\mathcal{F}$; ii) each probability of task $f$ to be requested, denoted as $P_f$, is assumed to be independent identical distributed (i.i.d.). Namely, $\sum_{f=1}^FP_f = 1$. In particular, we consider a homogeneous scenario, i.e., $P_f = \frac{1}{F}$. Since the local input data is generated in real time, its content may vary from time to time. However, the input local data size is assumed to be unchanged.

%For clarity, take the online RPG as an example. The map scenario information corresponds to the remote input data, the size of which is $I^S$. The history of all the players’ actions could provide a popularity distribution of the map preferences, i.e., $\left(P_f\right)_{f\in \mathcal{F}}$. The task space $\mathcal{F}$ is constructed according to the map information. The real-time player's pose information at the mobile device corresponds to the local input data, the size of which is $I^D$. Even though corresponding to each map information, the local input data may vary from time to time, but the size of the pose information is in general unchanged. The rendered view by combining the player's information with the map scenario information corresponds to the output data, the size of which is $O$. The overall latency corresponds to $\tau$.

\subsection{Caching and Computing Model}

First, consider the cache placement at the mobile device. From the above-mentioned task model, we can see that only caching of the remote input data can be considered. Denote with \textcolor{black}{$c_{f}\in \{0,1\}$} the caching decision \textcolor{black}{of task $f\in \mathcal{F}$}, where \textcolor{black}{$c_{f}=1$} means that the \textcolor{black}{remote input} data is cached at the mobile device and $c_f=0$, otherwise. Denote with $C$ (in bits) the cache size  at the mobile device and the caching constraint is given by
\begin{equation}\label{cache}
\sum_{f} I^{S}c_{f} \leq C.
\end{equation}
All the remote input data are assumed to be proactively cached at the MEC server considering the storage size at the MEC server is generally large enough. %In addition, from the perspective of the user's privacy, we assume the device-generated input data cannot be proactively cached at the MEC server.

Next, consider the computing decision at the mobile device. Denote with $d_f \in \{0,1\}$ the computing decision of task $f \in \mathcal{F}$, where $d_f=1$ means that task $f$ is computed at the mobile device and $d_f=0$ means that task $f$ is computed at the MEC server. Denote with $f_D$ (in cycles/second) the computation frequency of the mobile device and $f_S$ (in cycles/second) the computation frequency of the MEC server. The energy consumed for computing one cycle with frequency $f_D$ at the mobile device is $\mu f_D^2$, where $\mu$ is the effective switched capacitance related to the chip architecture and can indicate the power efficiency of CPU at the mobile device \cite{mao2}. Denote with $\bar{P}$ in (W) the average available power at the mobile device. We assume that there is no power constraint at the MEC server considering the MEC server is in general connected to a power grid.
 %Denin general conneote with $\mu$ thver is e computation efficiency, which is assumed to be the same at both the mobile device and  the MEC server.} 

%\subsection{Request Model}
% In addition, in order to satisfy the requirement of low latency in 5G era, we have the following time constraint. 
\subsection{Service Mechanism}
Based on the joint caching and computing decision $\left(\textbf{c}\triangleq \left(c_f\right)_{f\in \mathcal{F}},\ \textbf{d} \triangleq \left(d_f\right)_{f\in \mathcal{F}}\right)$, each task $f\in \mathcal{F}$ can be served via the following three routes. 
\begin{itemize}
    \item \textbf{Local computing with local caching}. When $d_f=1$ and $c_f = 1$, the mobile device immediately computes task $f$ based on the real-time local input data and the locally cached remote input data. The required latency is the computation latency at the mobile device only, i.e., $\frac{\left(I^S+I^D\right)w}{f_D}d_f$. For satisfying the latency constraint, we assume that $\frac{\left(I^S+I^D\right)w}{f_D}\leq \tau$.  The average consumed power at the mobile device for task $f$ is the consumed computation power only, i.e.,  $\frac{\mu f_D^2w\left(I^D+I^S\right)}{F\tau}d_f$. 
    \item \textbf{Local computing without local caching}. When $d_f=1$ and $c_f = 0$, the mobile device first downloads the remote input data  from the MEC server and then computes the task locally. The required latency includes the downloading latency and the local computation latency, i.e.,
    \begin{align}\label{latency1}
    \Big(\!\frac{I^{S}}{B_{f}^{D}\log(1+\frac{P_{D}h^2}{N_{0}})}\!+\!\frac{(I^{S}+I^{D})w}{f_D}\!\Big) d_f\!\left(1\!-\!c_f\right)\!\leq\! \tau,
\end{align}
where $B_f^D$ is the downlink bandwidth allocated for the transmission of task $f$, $P_{D}$ is the average downlink power spectrum density (PSD) at the MEC server, $h$ is the channel coefficient and $N_0$ is the average PSD of the channel noise. The average consumed power at the mobile device for task $f$ is the consumed computation power only, $\frac{\mu f_D^2w\left(I^D+I^S\right)}{F\tau}d_f$. 
    \item{\textbf{MEC computing}}. When $d_f=0$, the mobile device first uploads the local input data to the MEC server. After receiving the local input data, the MEC server computes task $f$ and then transmits the output data to the mobile device. The required latency includes the uplink transmission latency, the computation latency at the MEC server, and the downlink transmission latency, i.e., 
    \begin{align}\label{latency2}
    &\Bigg(\frac{I^{D}}{B_{f}^{U}\log(1+\frac{P_{U}h^2}{N_0})}+\frac{(I^{S}+I^{D})w}{f_S}\nonumber\\
    &+\frac{O_{f}}{B_f^{D}\log(1+\frac{P_{D}h^2}{N_{0}})}\Bigg)(1-d_f) \leq \tau,
\end{align}
where $B_f^U$ is the uplink bandwidth allocated to the mobile device and $P_{U}$ is the average uplink PSD at the mobile device. The average consumed power at the mobile device for task $f$ is the average uplink transmission power, i.e., $\frac{P_UI^{D}}{F\tau \log\left(1+\frac{P_{U}h^2}{N_{0}}\right)}(1-d_f)$. 
%In either case, it is necessary to keep the duration to be less than $\tau$. 
\end{itemize}
From above, under the average power constraint at the mobile device, we have
\begin{align}\label{power}
    \sum_{f=1}^F &\Bigg(\frac{\mu f_D^2w\left(I^D+I^S\right)}{F\tau}d_f \nonumber\\
    &+\frac{P_U I^{D}}{F\tau \log\left(1+\frac{P_{U}h^2}{N_{0}}\right)}(1-d_f)\Bigg)\leq \bar{P}.
\end{align}
The average consumed bandwidth, including both uplink and downlink bandwidth, is given by
\begin{align}
  \frac{1}{F}  \sum_{f=1}^F\left(B_f^U+B_f^D\right).
\end{align}

\section{Problem Formulation and Optimal Property Analysis}
\subsection{Problem Formulation}
We formulate the joint caching and computing optimization problem to minimize the average required bandwidth, including both the uplink and downlink bandwidth, subject to the cache size, average power and latency constraints, as below.
\begin{Prob}[Joint Caching and Computing Optimization]
\begin{align}
&\min_{\textbf{c},\ \textbf{d}}\ & \ \ \frac{1}{F}\sum_{f=1}^F (B_f^U+B_f^D)\nonumber\\
&\ s.t.\ & (\ref{cache}),\ (\ref{latency1}),\ (\ref{latency2}),\ (\ref{power}),\nonumber\\
&\ \ \ & c_f \in \{0,1\},\ d_f\in \{0,1\}, f\in \mathcal{F}.\nonumber
\end{align}
\end{Prob}
Denote with $\left(\textbf{c}^* \triangleq (c_f^*)_{f\in \mathcal{F}},\textbf{d}^*\triangleq (d_f^*)_{f\in \mathcal{F}}\right)$ the optimal joint caching and computing policy and $B^*$ the corresponding optimal average bandwidth.
\subsection{Optimal Properties}
First, we can directly observe the following property between the local computing and local caching.
\begin{property}
When $d_f=0$, $c_f = 0$ without loss of optimality. 
\end{property}
%Property~1 is obvious since when task $f$ is decided to be computed at the MEC server, caching the remote input data at the mobile device is a waste of storage space.

Then, for each $f\in \mathcal{F}$, introduce $x_{f,j}\in \{0,1\},\ j \in \{1,2,3\},$ with $x_{f,j} = 1$ indicating that task $f$ is served via the $j$-th route and $x_{f,j}=0$ otherwise. Here, the first route corresponds to the local computing with local caching, i.e., $d_f=1,\ c_f=1$. The second  refers to the local computing without caching, i.e., $d_f=1,\ c_f=0$. The third refers to the MEC computing, i.e., $d_f=0,\ c_f=0$. Denote with $B_{f,j},\ j\in \{1,2,3\}$ the minimum value of $B_f^U+B_f^D$ for the $j$-th route given $(\textbf{c},\ \textbf{d})$.

Next, under latency constraint, we obtain the analytical expression for $B_{f,j}$.
\begin{property}
\textcolor{black}{When $d_f=1,\ c_f=1$, i.e., $x_{f,1}=1$, \ $B_{f,1} = 0$. }
\end{property}
\begin{property}
\textcolor{black}{When $d_f=1,\ c_f=0$, i.e., $x_{f,2} = 1$,\ $B_{f,2} = \frac{I^S}{\left(\tau-\frac{(I^S+I^D)w}{f_D}\right)\log\left(1+\frac{P_Dh^2}{N_0}\right)}$.}
\end{property}
Property~3 can be obtained directly from (\ref{latency1}). 
\begin{property}
\textcolor{black}{When $d_f=0,\ c_f=0$, i.e., $x_{f,3} =1$,\ $B_{f,3} = \frac{(\sqrt{a_1}+\sqrt{a_2})^2}{a_3}$, where $a_1=\frac{I^D}{\log\left(1+\frac{P_Uh^2}{N_0}\right)}$, $a_2=\frac{O}{\log\left(1+\frac{P_D h^2}{N_0}\right)}$, and $a_3=\tau-\frac{(I^S+I^D)w}{f_S}$.}
\end{property}
\begin{proof}
\textcolor{black}{Proof of Property~4 can be seen in Appendix~A. }
\end{proof}

After that, via replacing $B_f^U+B_f^D$ in the objective function of Problem~1 with $B_{f,j}$ obtained from Properties 2-4, the latency constraints (\ref{latency1}) and (\ref{latency2}) can be eliminated. 
\textcolor{black}{Denote with $X_j \triangleq \sum_{f=1}^F x_{f,j},\ j\in \{1,2,3\}$ the number of tasks served via the $j$-th route. 
%$N=N_1+N_2$ where $N_1$ denotes the number of files following processing method 1 and $N_2$ denotes the number of files following processing method 2. 
Since each task is independent of each other and homogeneous, given $(X_j)_{j\in \{1,2,3\}}$, $(x_{f,j})_{f\in \mathcal{F},j\in \{1,2,3\}}$ can be obtained via 
\begin{equation}\label{service1}
x_{f,1} =
\begin{cases}
1& \text{$i = 1, \cdots, X_1$,}\\
0& \text{otherwise,}
\end{cases}
\end{equation}
\begin{equation}\label{service2}
x_{f,2} =
\begin{cases}
1& \text{$i = X_1+1, \cdots, X_1+X_2$,}\\
0& \text{otherwise,}
\end{cases}
\end{equation}
\begin{equation}\label{service3}
x_{f,3} =
\begin{cases}
1& \text{$i = X_1+X_2+1,\! \cdots,\! X_1+X_2+X_3$,}\\
0& \text{otherwise.}
\end{cases}
\end{equation}}

Via replacing $(x_{f,j})_{f\in \mathcal{F},j\in \{1,2,3\}}$ with $(X_j)_{j\in \{1,2,3\}}$, Problem~1 is transformed into Problem~2 equivalently. 

\textcolor{black}{\begin{Prob}[Equivalent Optimization]\label{prob:special2}
\begin{align}
&\min_{(X_j)_{j\in \{1,2,3\}}} X_1 B_1+X_2 B_2+X_3 B_3 \label{obj_special}\\
& ~~~~~~ s.t. ~~~~~\ \ I^S X_1 \leq C, \label{cache_special}\\
&    ~~~~~~~~\ \ \ \ \ \ \ \ \ k_{1}(X_1+X_2) + k_{2}X_3\leq \bar{P},\label{power_special}\\
& ~~~~~~~~~\ \ \ \ \ \ \ \  X_1+X_2+X_3 = F,\\
& ~~~~~~~~\ \ \ \ \ \ \ \ \ 0<=X_1<=F,\\
& ~~~~~~~~\ \ \ \ \ \ \ \ \ 0<=X_2<=F,\\
& ~~~~~~~~\ \ \ \ \ \ \ \ \ 0<=X_3<=F,
\end{align}
where $k_1 \triangleq \frac{\mu f_D^2w\left(I^D+I^S\right)}{\tau F}$ and $k_2 \triangleq \frac{P_U I^{D}}{F\tau \log\left(1+\frac{P_{U}h^2}{N_{0}}\right)}$ represent the average power consumed at the mobile device of each task for local computing and uplink transmission, respectively. 
\end{Prob}}

\section{Optimal Policy and Tradeoff Analysis}

\subsection{Optimal Policy}
\begin{theorem} \label{theorem1}
(Optimal joint policy when $k_1>k_2$) If $B_{3}>B_{2}$, the optimal joint policy is given as 
\begin{equation} \label{case1}
    \begin{aligned}
        %&X=\min\{F, floor(\frac{\bar{P}-Fk_2}{k_1-k_2})\}\\
        &X_{1}=\min \left \{ \left\lfloor  \frac{C}{I^{S}} \right \rfloor , F, \left \lfloor \frac{\bar{P}-Fk_2}{k_1-k_2} \right\rfloor \right \},\\
        &X_{2} = \max \left \{0, \min \left \{F,  \left \lfloor \frac{\bar{P}-Fk_2}{k_1-k_2} \right \rfloor \right \} -X_1 \right\},\\ %\nonumber
        & X_3 = F - X_1 - X_2,% \nonumber
        %&\mathbf{V}=\{X_1,X_2,F-X_1-X_2\},
    \end{aligned}
\end{equation}
where $\lfloor \bullet \rfloor$ denotes the round-down function.  $B^*=B_2X_2+B_3X_3 $. If $B_{3} \leq B_{2}$, 
the optimal joint policy is given as 
\begin{equation}\label{case2}
    \begin{aligned}
        &X_{1}=\min \left \{ \left \lfloor \frac{C}{I^{S}} \right \rfloor, F,  \left \lfloor \frac{\bar{P}-Fk_2}{k_1-k_2} \right \rfloor \right \},\\
        %&X_{1}=\min\{floor(C/I^{S}),X\}\\
        & X_{2} = 0, \\ 
        &X_3 = F -X_1.
    \end{aligned}
\end{equation}
$B^*=B_3X_3$.
\end{theorem}
\begin{proof}
Proof of Theorem~1 can be seen in Appendix~B.
\end{proof}

\begin{theorem}\label{theorem2}
(Optimal joint policy when $k_1 \leq k_2$) If $B_3>B_2$, the optimal joint policy is given as 
\begin{equation}\label{case3}
    \begin{aligned}
        &X_{1}= \left  \lfloor \frac{C}{I^S} \right \rfloor,\\
        &X_{2}=F-X_1,\\
        &X_3 = 0.
    \end{aligned}
\end{equation}
 $B^*=B_2X_2$. If $B_3 \leq B_2$, the optimal joint policy is %given as 
\begin{equation}\label{case4}
    \begin{aligned}
        &X_{1}=  \left \lfloor \frac{C}{I^S} \right \rfloor,\\
        &X_{2}=\max \left\{0, \left \lceil \frac{\bar{P}-F k_2}{k_2-k_1} \right \rceil-X_{1} \right\},\\
        &X_3 = F-X_1-X_2.
    \end{aligned}
\end{equation}
$B^*=B_2X_2+B_3X_3$.
\end{theorem}
\begin{proof}
Proof of Theorem~2 can be seen in Appendix~C.
\end{proof}

%As for constraint 1, it is obvious that, for more files, if $C_f^S=1$, then more corresponding bandwidth will be $B_f^D+B_f^U=0$. Moreover, since each file is the same and independent, we could just find the most number of $C_f^S$ which can take on value of 1, namely, $M = floor(\frac{C}{I^S})$, where $M$ denotes the maximum number of $C_f^S=1$.

\subsection{Tradeoff Analysis}

%In this section, we do simulation and analysis to verify the algorithm efficiency and optimality in terms of time, computation resource and totally utilized bandwidth. 

%Parameters 

%\begin{figure}[t]
%  \centering
%  \includegraphics[width=85mm, height=70mm]{Fig1.eps}
%\caption{$P_U=250mW/180KHz$, $P_D= 25W/180KHz$, $f_S=2*10^{11}$ Hz, $\mu=10^-27$ joule/cycle/$Hz^2$, $w=10$, $O=1.25$ MB, $I^D=1$ MB, $I^S=3.75$ MB, $d=10$ m, $N_0= 4*10^-21$ Watts/Hz, $f_S=15~GHz, f_D=[1.5 ~to~10]GHz, F=1000, C=1400~MB, \tau=0.8~s, d=10 m$} 
%\label{fig:fig1}
%\end{figure}
\begin{figure}[t]
\begin{center}
 \subfigure[\textcolor{black}{Cache size when $f_D=4$ GHz and $\tau=0.5$ s}.]
  {\resizebox{4.2cm}{!}{\includegraphics{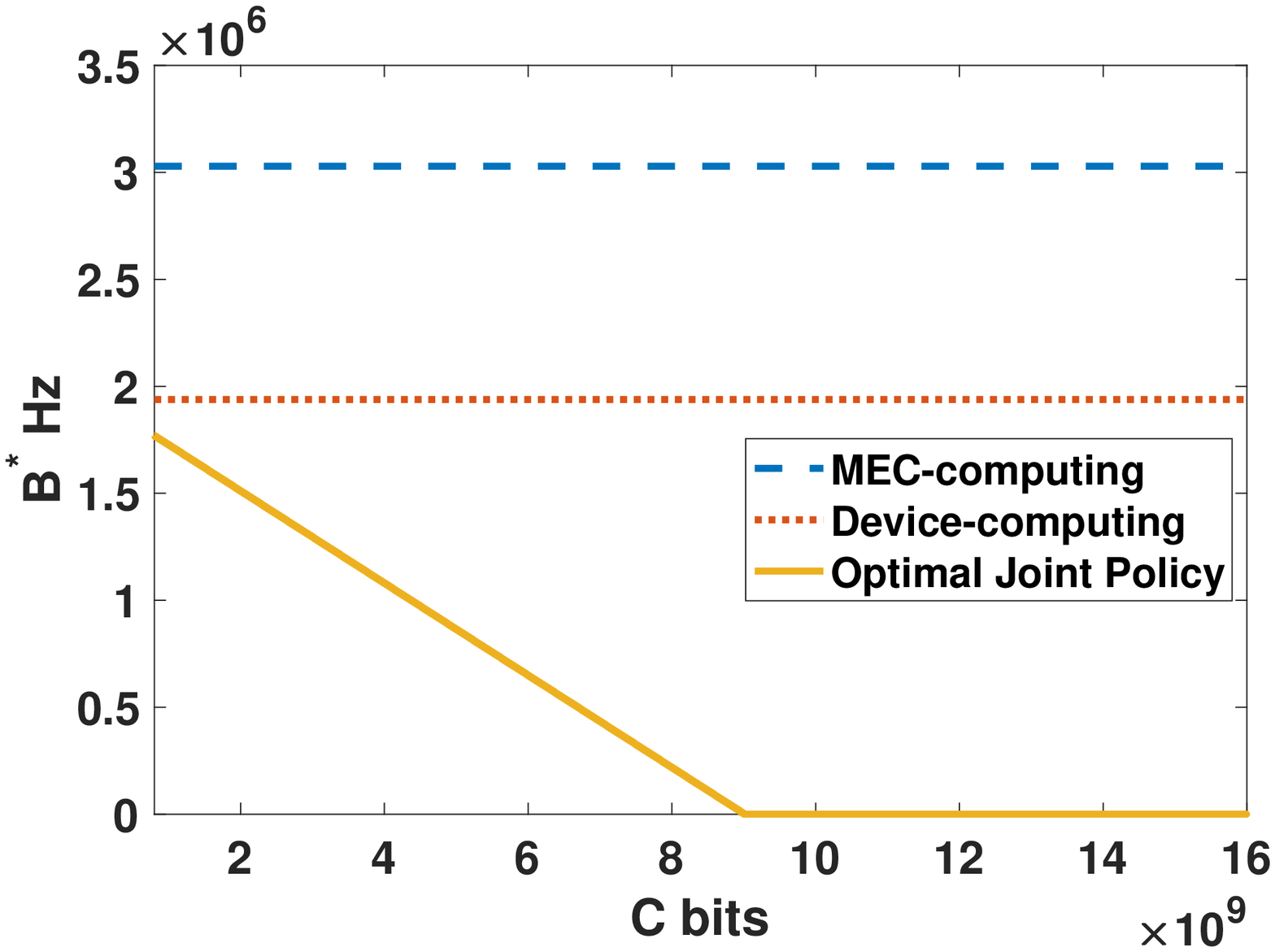}}}
   \subfigure[\textcolor{black}{\textcolor{black}{Computing} frequency when $C=400$ MB and $\tau=0.143$ s}.]
  {\resizebox{4.2cm}{!}{\includegraphics{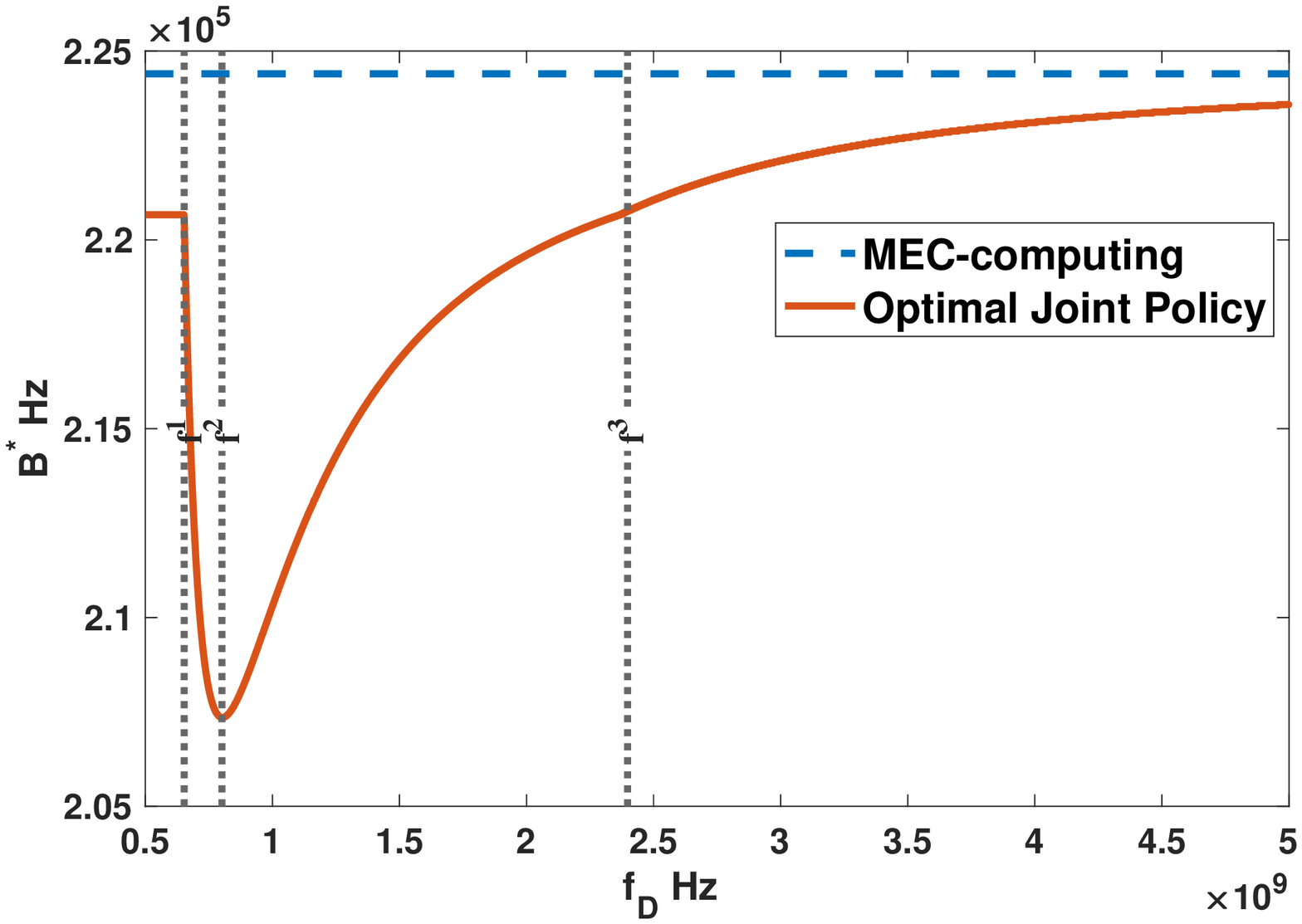}}}
 %{\resizebox{6.0cm}{!}{\includegraphics{F100_C.eps}}}\quad\quad
% {\includegraphics[height=1.5cm,width=3cm]{x.eps}}\quad\quad
% \subfigure[\textcolor{black}{Computation frequency \textcolor{black}{$f_V$ when $\bar{P}=7.5$ W} and $\frac{C}{N} = 30 \%$}.]
 %{\resizebox{6.4cm}{!}{\includegraphics{F100_K.eps}}}\quad\quad
% {\includegraphics[height=1.5cm,width=3cm]{x.eps}}\quad\quad
\end{center}
%\vspace{-8mm}
   \caption{\small{Impact of $C$ and $f_D$ on $B^*$.  $f_S=15$ GHz, $F=300$, $\frac{P_{D}h^2}{N_0}=28.1573$ dB, $\frac{P_{U}h^2}{N_0}=10.98$ dB, $P_{U}=\frac{250\ \text{mW}}{180\ \text{KHz}}$, $P_{D}=\frac{5\ \text{W}}{180\ \text{KHz}}$, $w=10$, $\mu=10^{-27}$.}}
\label{case2}
\end{figure}

\subsubsection{\textbf{$\mathbf{k_1>k_2}$}}
%Theorem \ref{theorem2} and \ref{theorem2} indicates four scenarios, which means we need to consider four cases for each parameter's relation to the total bandwidth.

%Firstly, let us investigate the property of $f_D$. 
%We hereby make a reasonable assumption that $(I^S+I^D)w<<f_D$ since the file size should be much smaller than the calculation frequency, otherwise the time efficiency is never a realizable term, namely, low latency in IoT/5G era. 

When $k_1>k_2$ and $B_3>B_2$, from (\ref{case1}),  there are three possible cases as below.
\begin{itemize}
\item When $ \left \lfloor \frac{\bar{P}-Fk_2}{k_1-k_2} \right \rfloor< \left \lfloor \frac{C}{I^{S}}\right \rfloor   <F  $,  $X_1 =  \left \lfloor \frac{\bar{P}-Fk_2}{k_1-k_2} \right \rfloor $ decreases with $f_D$ since $k_1$ increases with $f_D$, $X_2 = 0$ is independent of $f_D$ and $X_3 = F- X_1$ increases with $f_D$. Then, $B^* = B_3X_3$ increases with $f_D$ with $B_3$ given in Property~4 independent of $f_D$.  This is because when the locally available power is limited, increasing $f_D$ decreases the number of tasks that can be computed locally. Also, $B^*$ is independent of $C$ indicating that it is mainly limited by the local computing power $ \left \lfloor \frac{\bar{P}-Fk_2}{k_1-k_2} \right \rfloor$.

\item When $  \left \lfloor \frac{C}{I^{S}}\right \rfloor <  \left \lfloor \frac{\bar{P}-Fk_2}{k_1-k_2} \right \rfloor<F  $, $X_1 = \left \lfloor \frac{C}{I^{S}}\right \rfloor $ is independent of $f_D$, $X_2 =  \left \lfloor \frac{\bar{P}-Fk_2}{k_1-k_2} \right \rfloor -X_1$ decreases with $f_D$ since $k_1$ increases with $f_D$, and $X_3 $ increases with $f_D$.  \textcolor{black}{Then,  $B^*$ increases with $f_D$ since $B_2<B_3$. Meanwhile, $B^*$ decreases with $C$. }

\item When $ \left \lfloor \frac{C}{I^{S}}\right \rfloor   <F <\left \lfloor \frac{\bar{P}-Fk_2}{k_1-k_2} \right \rfloor  $,  $X_1 =  \left \lfloor \frac{C}{I^{S}}\right \rfloor$ and $X_2 = F - X_1$ are independent of $f_D$ and $X_3 = 0$ is independent of $f_D$. Then,  since $B_2$ given in Property~3 decreases with $f_D$, $B^* = B_2X_2$ decreases with $f_D$.  This is because when the locally available power is large enough, increasing $f_D$ decreases the computation latency. Also, $B^*$ decreases with $C$. 

\end{itemize}

%if $f_D$ increases, 

%then the consumed local computation power $k_1$ increases. Moreover, there are three possible modes, 1) $X$ decreases, $X_1$ unchanged, $X_2$ decreases, and $X_3$ increases if $X>=floor(C/I^{S})$ and $X<F$; 2) $X$ decreases, $X_1$ decreases, $X_2$ unchanged, and $X_3$ increases if $X<floor(C/I^{S})$ and $X<F$; 3) $X=F$ unchanged, $X_1$ unchanged, $X_2$ unchanged, and $X_3$ unchanged if $X>=F$.

When $k_1>k_2$ and $B_3 
\leq B_2$, from (17), there are three possible cases as below. 
\begin{itemize}
\item When $ \left \lfloor \frac{\bar{P}-Fk_2}{k_1-k_2} \right \rfloor< \left \lfloor \frac{C}{I^{S}}\right \rfloor   <F  $,  $X_1 =  \left \lfloor \frac{\bar{P}-Fk_2}{k_1-k_2} \right \rfloor $ decreases with $f_D$, $X_2 = 0$ is independent of $f_D$ and $X_3 = F- X_1$ increases with $f_D$. Then, $B^* = B_3X_3$ is independent of $C$ and increases with $f_D$.  This is because when the locally available power is limited, i.e., smaller than the number of tasks that can be cached locally $\frac{C}{I^S}$, increasing $f_D$ decreases the number of tasks that can be computed locally. 

\item When $  \left \lfloor \frac{C}{I^{S}}\right \rfloor <  \left \lfloor \frac{\bar{P}-Fk_2}{k_1-k_2} \right \rfloor<F  $, $X_1 = \left \lfloor \frac{C}{I^{S}}\right \rfloor $, $X_2 =  0$ and $X_3 = F - X_1 $ which are all independent of  $f_D$.  Then, we have $B^* =  B_3X_3$ independent of $f_D$ and decreases with $C$. This is because when $  \left \lfloor \frac{C}{I^{S}}\right \rfloor <  \left \lfloor \frac{\bar{P}-Fk_2}{k_1-k_2} \right \rfloor<F  $, the bandwidth gain is limited by the local cache size $C$.

\item When $ \left \lfloor \frac{C}{I^{S}}\right \rfloor   <F <\left \lfloor \frac{\bar{P}-Fk_2}{k_1-k_2} \right \rfloor  $,  $X_1 =  \left \lfloor \frac{C}{I^{S}}\right \rfloor$, $X_2 = 0$  and $X_3 = F-X_1$ are independent of $f_D$. Then, $B^* $ is independent of $f_D$ and decreases with $C$. %since $B_2$ given in Property~3 decreases with $f_D$.  This is because when the locally available power is large enough, increasing $f_D$ decreases the local computation delay. 

\end{itemize}

%if $f_D$ increases, then $k_1$ increases. Moreover, there are three possible modes, 4) $X$ decreases, $X_1$ unchanged, $X_2=0$, and $X_3$ unchanged if $X>=floor(C/I^{S})$ and $X<F$; 5) $X$ decreases, $X_1$ decreases, $X_2=0$, and $X_3$ increases if $X<floor(C/I^{S})$ and $X<F$; 6) $X=F$ unchanged, $X_1$ unchanged, $X_2$ unchanged, and $X_3$ unchanged if $X>=F$.

From Fig.~2 (a), we can see that joint caching and computing at the mobile device helps further reduce the bandwidth compared with computing only either at the MEC server or at the mobile device. From Fig.~2 (b), firstly, when $f_D$ is relatively small, $B_2 \geq B_3$, $k_1>k_2$ and $ \left \lfloor \frac{C}{I^{S}}\right \rfloor < F<\left \lfloor \frac{\bar{P}-Fk_2}{k_1-k_2} \right \rfloor$, and thus bandwidth remains unchanged with $f_D$. The first turning point $f^1$ appears when $B_3=B_2$. By setting $B_3=B_2$, the switching point $f^1$ could be explicitly expressed as
\begin{equation}
    f^1=\frac{(I^S+I^D)w}{\tau-\frac{a_3 I^S }{\log(1+\frac{P_D h^2}{N_0})(\sqrt{a_1}+\sqrt{a_2})^2}}.
\end{equation}

Then, the bandwidth starts decreasing with $f_D$. This is because as $f_D$ increases, $B_2 < B_3$, $k_1>k_2$ and $ \left \lfloor \frac{C}{I^{S}}\right \rfloor <F<\left \lfloor \frac{\bar{P}-Fk_2}{k_1-k_2} \right \rfloor $.  Then, the second turning point $f^2$ appears when $ \left \lfloor \frac{\bar{P}-Fk_2}{k_1-k_2} \right \rfloor =F$. By setting $ \left \lfloor \frac{\bar{P}-Fk_2}{k_1-k_2} \right \rfloor =F$, we could obtain the explicit expression for the second turning point
\begin{equation}
    f^2 \approx \sqrt{\frac{\tau(\bar{P}-Fk_2)+\tau F k_2}{\mu w (I^D+I^S)}}.
\end{equation}

Next, the bandwidth $B^*$ starts increasing with $f_D$. This is because as $f_D$ increases, $ \left \lfloor \frac{C}{I^{S}}\right \rfloor <\left \lfloor \frac{\bar{P}-Fk_2}{k_1-k_2} \right \rfloor <F$. Moreover, we could observe that there is another turning point which pushes the optimal policy towards the bandwidth of MEC-computing policy eventually. This turning point happens when $ \left \lfloor \frac{C}{I^{S}}\right \rfloor =\left \lfloor \frac{\bar{P}-Fk_2}{k_1-k_2} \right \rfloor$. By setting $ \left \lfloor \frac{C}{I^{S}}\right \rfloor =\left \lfloor \frac{\bar{P}-Fk_2}{k_1-k_2} \right \rfloor$, the turning point $f^3$ can be expressed as 
\begin{equation}
    f^3 \approx \sqrt{\frac{\tau F I^S(\bar{P}-Fk_2)}{\mu w(I^S+I^D)C}+\frac{\tau F k_2}{\mu w (I^S+I^D)}}.
\end{equation}
When $f_D>f^3$, the optimal policy is the scenario where $k_1>k_2$, $B_3>B_2$ and $ \left \lfloor \frac{\bar{P}-Fk_2}{k_1-k_2} \right \rfloor< \left \lfloor \frac{C}{I^{S}}\right \rfloor   <F  $. The optimal policy converges to the MEC computing policy as $f_D$ goes to infinity, i.e., $X_3 = F$. %This phenomenon aligns , the $X_3$ will keep increasing until it reaches $F$. 

%Above analysis shows that the mode transferring and the corresponding derivation of the turning point. Next, we complete all the mode existing with $f_D$ being the variable and how $X_1$, $X_2$, and $X_3$ changes when $k_1 <= k_2$.

\subsubsection{\textbf{$\mathbf{k_1 \leq k_2}$}}
When $k_1 \leq k_2$ and $B_3>B_2$, from (\ref{case3}), there is only one possible case. The bandwidth gain mainly comes from the local computing with/without caching. The MEC computing does not bring any gain.

\section{Conclusion}
In this paper, we consider a novel bidirectional computation task model and formulate the joint caching and computing optimization problem to minimize the average bandwidth under the latency, cache size and average power constraints. We derive the closed-form expressions for the optimal policy and  the minimum bandwidth, which illustrates that the 3C tradeoff can be classified into nine regions according to the relationship between the cache and computation capabilities at the mobile device, that between the uplink transmission power consumption and the local computation power consumption.%, as well as that between the bandwidth requirement via the local computing without caching and that via the MEC computing. %\textcolor{black}{In the heterogeneous scenario, we propose the LR+CCCP method and show that it outperforms the state-of-the-art CCCP and the heuristic greedy algorithm from both the bandwidth performance and the time efficiency.}

%The optimal policy is the one minimizing the bandwidth for file transmission. We firstly come up with the solution for the homogeneous scenario in which all service-related file size is the same and any file request is deemed as i.i.d.. The simulation has been implemented to show the superiority of the solution over pure MEC server computing or pure mobile device computing without a dynamic task offloading scheme. Next, the heterogeneous scenario is considered as a general case. This problem is proven to be an MMKP problem. 
\section*{Appendix A: Proof of Property~4}
For each task $f\in \mathcal{F}$, when $d_f=0$, we have  $\frac{I_{f}^{D}}{B_{f}^{U}\log(1+\frac{P_{U}h^2}{N_{0}})}+\frac{(I_{f}^{S}+I_{f}^{D})w_f}{f_S}+\frac{O_{f}}{B_f^{D}\log(1+\frac{P_{1}h^2}{N_{0}})} \leq \tau$. Hence, $B_{f,3}$ can be obtained via solving the following optimization problem:
\begin{equation}
\begin{aligned}
&\min_{B_{f}^U,B_{f}^D}{~~~B_f^U+B_f^D} \\
&\ \ s.t.~~~~\frac{a_1}{B_f^U}+\frac{a_2}{B_f^D}\leq a_3,\\
    &~~~~~~~~~~~~B_f^U > 0,\\
    &~~~~~~~~~~~~B_f^D > 0,
\end{aligned}
\end{equation}
where $a_1=\frac{I_{f}^{D}}{\log(1+\frac{P_{U}h^2}{N_{0}})} > 0$, $a_2=\frac{(I_{f}^{S}+I_{f}^{D})w_f}{f_0} > 0$ and $a_3=\tau-\frac{O_{f}}{\log(1+\frac{P_{D}h^2}{N_{0}})} > 0$. We can see that Problem~10 is a convex minimization problem. Denote with $B_f^{U^*}$ and $B_f^{D^*}$ the optimal solution to Problem~10. In order to solve Problem~10, let us first consider a modified version of the above convex problem as below.
\begin{equation}
\begin{aligned}
&\min_{B_{f}^U,B_{f}^D}{~~~B_f^U+B_f^D} \\
&\ \ s.t.~~~~\frac{a_1}{B_f^U}+\frac{a_2}{B_f^D}\leq a_3.\\
\end{aligned}
\end{equation}
If the solution to Problem 11 satisfies $B_f^U>0$ and $B_f^D>0$, then it is also a solution to Problem~10. Based on KKT conditions of Problem~11, we get an optimal solution to Problem~11 as below.
\begin{align}
    &B_{f}^{U^*} = \frac{a_1+\sqrt{a_1 a_2}}{a_3}>0\\
    &B_{f}^{D^*} = \frac{a_2+\sqrt{a_1a_2}}{a_3}>0.
\end{align}
Therefore, we get $B_f^{U^*}$ and $B_f^{D^*}$ of Problem~10, and then $B_{f,3} = B_f^{U^*}+B_f^{D^*}$. The proof ends.

\section*{Appendix B: Proof of Theorem \ref{theorem1}}

\begin{itemize}
\item{Suppose $B_3>B_2$ and $k_1-k_2>0$, from Problem. \eqref{obj_special} constraint \eqref{power_special}, we could obtain an upper-bound, $X_1+X_2<=\lfloor \frac{\bar{P} -Fk_2}{k_1-k_2} \rfloor$. Meanwhile, constraint. \eqref{cache_special} yields upper-bound $X_1<=\lfloor \frac{C}{I^S} \rfloor$. Since $B_1=0$, assigning as much files to processing method 1 as possible is the best policy. Meanwhile, $X_1$ cannot be larger than the total number of files $F$ obviously.Therefore, we have $X_1=\min\{ \lfloor \frac{C}{I^S} \rfloor, F, \lfloor \frac{\bar{P} -Fk_2}{k_1-k_2} \rfloor \}$. Subsequently, because $B_3>B_2$, assigning files to processing method 2 is the optimal policy with the constraint of $X_1+X_2 <= \lfloor \frac{\bar{P} -Fk_2}{k_1-k_2} \rfloor$ which could be larger than the total number of files. Hence, $X_2=\max\{0,min\{F,\lfloor \frac{\bar{P} -Fk_2}{k_1-k_2} \rfloor \}-X_1\}$. Last but not least, $X_3=F-X_1-X_2$.}

\item{If $B_3<=B_2$ and $k_1-k_2>0$, from Problem. \eqref{obj_special} constraint \eqref{power_special}, we could obtain an upper-bound, $X_1+X_2<=\lfloor \frac{\bar{P} -Fk_2}{k_1-k_2} \rfloor$. Meanwhile, constraint \eqref{cache_special} yields upper-bound $X_1<=\lfloor \frac{C}{I^S} \rfloor$. Since $B_1=0$, assigning as much files to processing method 1 as possible is the best policy. Therefore, similarly with the above proof, we have $X_1=\min\{ \lfloor \frac{C}{I^S} \rfloor, F, \lfloor \frac{\bar{P} -Fk_2}{k_1-k_2} \rfloor \}$. Because $B_3<=B_2$, assigning files to processing method 3 is the optimal policy. So we do not utilize process method at all. Hence, $X_2=0$ and $X_3=F-X_1$.}
\\\\
The proof ends here.

\end{itemize}

\section*{Appendix C: Proof of Theorem \ref{theorem2}}

\begin{itemize}
\item{
Suppose $B_3>B_2$ and $k_1-k_2<=0$. Constraint \eqref{power_special} yields that 
$X_1+X_2 >= \lceil \frac{\bar{P}-k_2F}{k_1-k_2} \rceil$.Meanwhile, constraint \eqref{cache_special} yields $X_1<=\lfloor \frac{C}{I^S} \rfloor$. Therefore, we have only one upper-bound and we set $X_1=\lfloor \frac{C}{I^S} \rfloor$. Since $B_3>B_2$, we set $X_2=F-X_1$ and $X_3=0$.}
\item{
Suppose $B_3<=B_2$ and $k_1-k_2<=0$. Similarly, one upper-bound indicates that $X_1=\lfloor \frac{C}{I^S} \rfloor$. Since $B_3<=B_2$, we get rid of $X_2$ as much as possible by setting $X_2=\max\{0,\lceil \frac{\bar{P}-F k_2}{k_1-k_2}) \rceil -X_{1}\}$.  Note that it is possible for $\frac{\bar{P}-k_2F}{k_1-k_2}$ to be negative so that there is no constraint for $X$, which means there is no constraint on $X_2$ and it could be zero directly. Last, $X_3=F-X_1-X_2$.}
\\\\
The proof ends here.
\end{itemize}

\bibliographystyle{IEEEtran}
\bibliography{reference}
\end{document}